\documentclass[12pt]{article}

\usepackage{amsmath,amssymb,amsthm,mathtools}
\usepackage{enumitem}
\usepackage{graphicx}
\usepackage{hyperref}

\usepackage{setspace}
\newtheorem{theorem}{Theorem}[section]

\newtheorem{proposition}[theorem]{Proposition}
\newtheorem{corollary}[theorem]{Corollary}

\theoremstyle{definition}
\newtheorem{definition}[theorem]{Definition}
\newtheorem{remark}[theorem]{Remark}

\newcommand{\eg}{{e.g.}, }
\newcommand{\R}{\mathbb{R}}
\newcommand{\N}{\mathbb{N}}
\newcommand{\E}{\mathbb{E}}
\newcommand{\IF}{\mathrm{IF}}
\newcommand{\GES}{\gamma^{*}}

\usepackage{fancyhdr}
\title{\Large Weak Moment Methods for Statistical Inference\\
\large with an Application to Robust Estimation}
\author{ R. Labouriau}
\date{Spring 2026}

\begin{document}
\maketitle

\setstretch{0.9}
\vspace{-1cm}
\begin{abstract}
\fontsize{8}{12}\selectfont
A companion paper~\cite{LabouriauA1} develops a generalised
probabilistic framework in which a probability law is represented by a
tempered distribution $T\in\mathcal{S}'(\R)$---on the same footing as a
density or a characteristic function---and information is extracted by
pairing $T$ with a positive Schwartz kernel $\varphi\in\mathcal{S}(\R)$
that acts as a measurement instrument rather than as part of the law;
it shows that the resulting weak moments of all orders exist
unconditionally.  The
present paper turns this structure into a concrete methodology for
statistical inference.

We develop estimation strategies based on weak moments, weak
characteristic functions, weak cumulants, and regularised
reconstruction of the underlying density via Tikhonov inversion of the
multiplication operator $M_\varphi\colon f\mapsto f\varphi$.  A key
feature of this programme is that parametric inference proceeds
directly from weak expectations---moments, transforms, or
cumulants---without requiring reconstruction of the underlying
density; reconstruction is an additional route, useful when
density-level inference is the goal.  Each strategy extends a
classical procedure to heavy-tailed models in which the corresponding
classical object does not exist.

The central contribution is to show that weak moment estimators are
\emph{automatically} locally robust in the sense of Hampel.  The score
function of a weak moment estimator has the form
$x^j\varphi(x)$ minus the corresponding theoretical weak moment; it
is bounded and redescending for every positive Schwartz kernel, its
influence function admits a closed form, and its gross error
sensitivity is finite in every identifiable parametric model---without
ad hoc truncation.  The kernel thus plays the role of Huber's tuning
constant, but instead of a post-hoc modification of the estimator it
forms part of the chosen characterisation of the model---the
prescription of how the law is observed.  The arbitrariness inherent in
any such tuning is thereby not removed but relocated to a more
model-related---and, we conjecture, in some cases more natural---part of
the inference.

The programme is worked out in detail for the Cauchy location model
(where no classical moment estimator exists), for a Student
$t_3$ location--scale model, for a bivariate Cauchy location
model illustrating the multivariate extension, and for a bivariate
$t_3$ location--scale model that reveals the breakdown of the MLE
scale estimate under contamination.  In each case we compute the
influence function, the gross error sensitivity, and the asymptotic
variance in closed form, and compare the resulting estimators
numerically with the median, the MLE, the Huber $M$-estimator, and
the Tukey biweight, under both the correctly specified model and
under contamination.
Although the present paper focuses on parametric models, the
reconstruction route is inherently non-parametric and opens a path
to weak density estimation without parametric assumptions; we outline
this direction in Section~\ref{sec:discussion}.
\end{abstract}

\newpage
\tableofcontents
\newpage
\setstretch{1.1}

\section{Introduction}
\label{sec:introduction}

Moment-based methods are among the oldest tools in statistics, yet they
fail for heavy-tailed models---Cauchy, low-degree Student~$t$, stable
laws---whose classical moments do not exist.
In a companion paper~\cite{LabouriauA1} we developed a framework in which
a probability law is carried by a tempered distribution $T$ and probed
by a Schwartz kernel $\varphi$---the kernel an instrument for extracting
information, not part of the law---with expectations defined via
$\E_{T,\varphi}[\psi]=\langle T,\psi\varphi\rangle$.  Because $\varphi\in\mathcal{S}(\R)$ ensures
$x^j\varphi\in\mathcal{S}(\R)$, \emph{weak moments}
${}^{(\varphi)}m_j=\langle T,x^j\varphi\rangle$ exist for every $j$ and
every pair.  Labouriau~\cite{LabouriauA1} develops the algebra
(additivity of weak cumulants, a weak CLT) and the analysis (a
hierarchy of uniqueness theorems for the weak moment problem).  It
closes with a minimal illustration: even the first weak moment yields a
consistent estimator for the Cauchy location parameter.

The present paper turns that illustration into a full methodology.  Its
two objectives are: (i)~to develop estimation procedures---weak moment
matching, transform-based and cumulant-based methods, regularised
reconstruction of the density---inside the weak framework; and (ii)~to
show that these estimators are \emph{automatically} locally robust:
their influence function is bounded, their gross error sensitivity
finite, and their score redescending, all inherited from the decay of
the kernel~$\varphi$ with no ad hoc truncation.  A key feature is that
parametric inference proceeds directly from weak expectations without
requiring reconstruction of the underlying density; reconstruction is
an additional route, useful when density-level inference is the goal.
Although the present paper focuses on parametric models, the
reconstruction pathway is inherently non-parametric and opens a
direction toward weak density estimation without parametric assumptions
(see Section~\ref{sec:discussion}).  The central methodological point
is that inference in the weak framework is fundamentally direct:
parameters are identified and estimated from weak expectations without
reconstructing the underlying density; reconstruction is only required
when the density itself is the inferential target---a secondary,
self-contained route in this paper, separable from its direct-inference
core and a natural candidate for separate development.

\medskip

\emph{Beyond densities and likelihoods.}
Two features of the distributional representation deserve emphasis at
the outset.  First, it does not presuppose a density: because the law
is carried by a tempered distribution~$T$, the construction applies
unchanged to models that possess no density, and hence no
likelihood---non-dominated families, such as a moving atom superimposed
on a continuous background, for which likelihood-based inference has no
starting point.  Second, the representation is structurally stable: by
the classical structure theorem every tempered distribution is a finite
sum of derivatives of ordinary (polynomially bounded, continuous)
functions (Strichartz~\cite{Strichartz2003}, \S6.3;
cf.~\cite[Remark~2.2]{LabouriauA1}), so singular laws---point
masses, jumps---are differentiated regular functions rather than
pathologies, and the kernel returns them to bounded, smooth weak
quantities.  Since the tempered distribution~$T$ characterises the law,
this structural stability is a property of the model itself---of the
probability measures in play---and not of the estimation method.

\medskip

\emph{Why robustness is the natural home of weak moment methods.}
At first sight, weak moment estimators might appear to be merely a
device for handling heavy-tailed models.  We argue that the connection
to robust statistics is deeper.  First, in the classical programme of
Hampel and Huber~\cite{Hampel1974,Huber1964}, one starts from an
efficient but fragile score (typically the MLE score) and
\emph{truncates} or \emph{redescends} it to bound the influence
function.  The resulting estimator depends on a tuning constant.  In
the weak framework the score
$x^j\varphi(x)-{}^{(\varphi)}m_j(\theta)$ is automatically bounded
and redescending, and the ``tuning'' is the kernel~$\varphi$---part of
the chosen characterisation of the model, not a post-hoc modification.
Second, finite gross error sensitivity holds for essentially any
positive Schwartz kernel and any parametric family, regardless of the
tail behaviour of~$f_\theta$.  In particular, the framework produces
locally robust estimators for the Cauchy distribution, where the
classical Hampel programme has no efficient starting point.  Third,
the weak framework supplies a \emph{family} of estimators parametrised
by the moment order~$j$, the kernel~$\varphi$, and GMM weighting.
Optimising over this family is the analogue of the classical
Hampel-optimal estimation problem.

\medskip

\emph{Notation.}
Following~\cite{LabouriauA1}, weak objects carry a
superscript~$\varphi$:
${}^{(\varphi)}m_j$ for weak moments, ${}^{(\varphi)}\phi(t)$ for the weak
characteristic function, ${}^{(\varphi)}K(t)$ for the weak cumulant
generating function, ${}^{(\varphi)}\kappa_j$ for weak cumulants.  Classical
objects carry no superscript.  The superscript is a label, not a power.

\medskip

\emph{Organisation.}
Section~\ref{sec:setup} recalls the minimal setup
from~\cite{LabouriauA1} and introduces parametric weak models with
their empirical counterparts.
Section~\ref{sec:estimation} develops the estimation strategies:
direct methods (weak moment matching, transform-based and cumulant
methods) and density reconstruction via regularised inversion.
Section~\ref{sec:robust}---the centrepiece---places weak estimators
inside the theory of robust statistics.
Section~\ref{sec:examples} works out the Cauchy and Student~$t_3$
examples in detail and illustrates the multivariate extension through
elliptically contoured models.  Section~\ref{sec:simulation} presents
four Monte Carlo studies---univariate Cauchy location, univariate
$t_3$ location--scale, bivariate Cauchy location, and bivariate
$t_3$ location--scale---comparing weak moment estimators with
classical and robust benchmarks under both the clean model and
contamination.  Section~\ref{sec:discussion} concludes.

\section{Setup: parametric weak models and empirical functionals}
\label{sec:setup}

We assume the reader is familiar with the framework
of~\cite{LabouriauA1}; this section fixes notation and records,
without proofs, the properties we shall use.

\subsection{Distribution--kernel pairs}

A \emph{distribution--kernel pair} is $(T,\varphi)$ with
$T\in\mathcal{S}'(\R)$ and $\varphi\in\mathcal{S}(\R)$; the
\emph{weak expectation} of a polynomially bounded $\psi$ with
$\psi\varphi\in\mathcal{S}(\R)$ is
$\E_{T,\varphi}[\psi]:=\langle T,\psi\varphi\rangle$.
Throughout, $T$ represents the probability law and $\varphi$ is the
fixed kernel through which it is measured; following~\cite{LabouriauA1}
the kernel is an instrument, not part of the law, so every weak
quantity below is a measurement of the single underlying law~$T$.
In the case that a probability density exists and is sufficiently
regular ($T=T_f$, $f\ge 0$),
$\E_{T_f,\varphi}[\psi]=\int\psi(x)\varphi(x)f(x)\,dx$.
The weak moments ${}^{(\varphi)}m_j:=\E_{T,\varphi}[x^j]$, the weak
characteristic function
${}^{(\varphi)}\phi(t):=\E_{T,\varphi}[e^{itx}]$, and the weak cumulants
${}^{(\varphi)}\kappa_j:=(1/i^j)(d^j/dt^j)\log{}^{(\varphi)}\phi(t)|_{t=0}$
are all well defined for every pair and every $j$.  The uniqueness
theorems of~\cite{LabouriauA1} guarantee that the weak moment sequence
determines the underlying distribution under mild conditions
on~$\varphi$ (Gaussian kernels: via Hermite completeness; positive
Schwartz kernels: via a Carleman condition; exponential-decay kernels:
via Denjoy--Carleman quasi-analyticity).  In the density case, weak
expectations also determine the kernel-weighted distribution function
and the regularised density $g=f\varphi$
(see~\cite[Proposition~2.12 and Appendix~C]{LabouriauA1}); this fact
underpins the reconstruction strategy of
Section~\ref{subsec:reconstruction}.

\begin{definition}[Generalised random variable]\label{def:grv}
A \emph{generalised random variable} is a random object whose
distribution is specified by a distribution--kernel pair
$(T,\varphi)$, and whose expectations are defined by
$\E_{T,\varphi}[\psi]=\langle T,\psi\varphi\rangle$.  See
\cite[Definition~2.8]{LabouriauA1} for a full treatment.
\end{definition}

\begin{remark}[Interpretation in the density case]\label{rem:grv-density}
When $(T,\varphi)$ arises from a density~$f$, a generalised random
variable can be understood as an ordinary random variable $X\sim f$
observed through the kernel~$\varphi$: the observations
$X_1,\ldots,X_n$ are in practice ordinary draws from~$f$, and the
kernel enters only through the inferential functionals
$\psi(X_i)\varphi(X_i)$.  That is, the ``generalised'' character
resides in the expectation operator, not in the sampling mechanism.
\end{remark}

Thus, \cite{LabouriauA1} establishes the probabilistic and analytic
foundations of the weak framework, while the present paper develops
its statistical and inferential consequences.

\subsection{Parametric weak models}

Let $\Theta\subseteq\R^p$ be open, let $\varphi\in\mathcal{S}(\R)$
with $\varphi>0$, and let $\{T_\theta\}_{\theta\in\Theta}$ be a
family of tempered distributions such that each $(T_\theta,\varphi)$
is a probability pair.  We call
$\{(T_\theta,\varphi)\}_{\theta\in\Theta}$ a \emph{parametric weak
model}.  In the density case the weak moments are
${}^{(\varphi)}m_j(\theta)=\int x^j\varphi(x)f_\theta(x)\,dx$.

The kernel is an instrument for extracting information about the model
from the tempered distribution~$T$: different kernels yield different
weak moments and hence different estimators.  Structural
identifiability---injectivity of $\theta\mapsto T_\theta$---is a
property of the distributional family and is preserved by any positive
kernel under the uniqueness results of~\cite{LabouriauA1}.  In
estimation, however, one works with a finite collection of weak
moments or transforms, and local identifiability is expressed through
the full-rank condition on the Jacobian $G(\theta)$ of
Proposition~\ref{prop:wm-asy}.  The kernel therefore affects not only
efficiency but also the operative identifiability of the estimating
equations.

\begin{remark}[Three levels of identifiability]\label{rem:ident}
Three notions of identifiability are relevant in this framework:
(i)~structural identifiability of $\{T_\theta\}$, a property of the
distributional family alone;
(ii)~identifiability from the full weak moment sequence, which under
positive kernels inherits the uniqueness theorems
of~\cite{LabouriauA1}; and
(iii)~local identifiability from the chosen estimating equations,
governed by $\operatorname{rank}G(\theta)$.  Level~(i) is
model-intrinsic; level~(ii) connects to the weak moment problem;
level~(iii) is procedure-dependent and involves the kernel through the
moment map $\theta\mapsto{}^{(\varphi)}m_j(\theta)$.
\end{remark}

\subsection{Empirical weak expectations}
\label{subsec:empirical}

Let $X_1,\ldots,X_n$ be i.i.d.\ generalised random variables with
distribution $(T_\theta,\varphi)$ in the sense of
Definition~\ref{def:grv}.  The \emph{empirical weak expectation} is
\begin{equation}\label{eq:emp-weak}
  \widehat{\E}_{}^{(\varphi)}n[\psi]
  :=\frac{1}{n}\sum_{i=1}^n\psi(X_i)\,\varphi(X_i).
\end{equation}
The factor $\varphi(X_i)$ inside the average is essential: it matches
the theoretical weak expectation
$\int\psi\varphi f_\theta\,dx$, ensures that $\psi\varphi$ is bounded
(hence in $L^2(f_\theta)$), and produces automatic downweighting of
observations far from the kernel's centre.  By the law of large numbers,
$\widehat{\E}_{}^{(\varphi)}n[\psi]\to\E_{T_\theta,\varphi}[\psi]$ a.s.,
and the CLT gives asymptotic normality with variance
$\mathrm{Var}_{f_\theta}(\psi\varphi)
=\E_{T_\theta,\varphi^2}[\psi^2]-(\E_{T_\theta,\varphi}[\psi])^2$,
which is finite unconditionally.

The general estimation principle is therefore: match empirical and
theoretical weak expectations,
$\widehat{\E}_{}^{(\varphi)}n[\psi]\approx\E_{T_\theta,\varphi}[\psi]$,
over a family of test functions $\psi$.  Specialising to
$\psi_j(x)=x^j$ gives weak moment estimators; to $\psi_t(x)=e^{itx}$
gives transform-based estimators; to derivatives of
$\log\phi_{}^{(\varphi)}\theta$ gives weak cumulant methods.

\subsection{Approximation and observational resolution}
\label{subsec:resolution}

It is useful to read the kernel as a level of \emph{observational
resolution}.  The weak expectation $\E_{T,\varphi}[\psi]=\langle
T,\psi\varphi\rangle$ is a measurement of the law~$T$ at the resolution
set by~$\varphi$: the kernel localises the probe~$\psi$ to the region
where $\varphi$ is appreciable and damps its tails, so that even
unbounded or non-integrable features of~$T$ are seen through a finite
aperture.  Two limiting regimes make this reading precise.

\emph{Stability.}
Tempered distributions act continuously on Schwartz space, so the
measurement is stable under smooth changes of the instrument: if
$\varphi_j\to\varphi$ in $\mathcal{S}(\R)$ then
$\langle T,\psi\varphi_j\rangle\to\langle T,\psi\varphi\rangle$
(\cite[Remark~2.22]{LabouriauA1}).  Weak moments, weak characteristic
functions, and the estimators built from them therefore depend
continuously on the observational kernel; the dependence on~$\varphi$
is not a formal artefact but a controlled, intrinsically regularised
one.

\emph{Concentration.}
If the kernel is replaced by an approximate identity
$\rho_{\varepsilon}(x)=\varepsilon^{-1}\rho((x-x_0)/\varepsilon)$ with
$\rho\in\mathcal{S}(\R)$ and $\rho_\varepsilon\to\delta_{x_0}$, the
measurement concentrates at~$x_0$; in the density case
$\E_{T_f,\rho_\varepsilon}[\psi]\to\psi(x_0)f(x_0)$ at points of
continuity, recovering pointwise values of the weighted density.

\emph{Recovery of the classical object.}
Sending the kernel to the constant~$1$ removes the aperture.  For the
Gaussian kernel $\varphi_\sigma(x)=e^{-x^2/(2\sigma^2)}$ this is the
limit $\sigma\to\infty$, and whenever the classical expectation
$\int\psi\,dF$ exists one recovers it,
$\E_{T_f,\varphi_\sigma}[\psi]\to\int\psi\,dF$, by dominated
convergence; the general $\varphi\to1$ statement
is~\cite[Lemma~2.10]{LabouriauA1}.  Inferentially, the weak estimators
of Section~\ref{sec:estimation} are therefore classical estimators
observed at finite resolution: as $\sigma\to\infty$ they reduce to
their classical counterparts---the sample mean, the empirical
characteristic function---wherever those exist, and they remain well
defined (bounded score, redescending influence, finite variance)
precisely in the heavy-tailed and singular cases where the classical
objects break down.  The bandwidth~$\sigma$ thus interpolates between
classical efficiency (large~$\sigma$) and robustness (small~$\sigma$),
a trade-off taken up in Section~\ref{sec:robust}.

\begin{remark}[Structure and a microlocal perspective]
\label{rem:structure-microlocal}
The resolution reading meshes with the fine structure of tempered
distributions.  By the structure theorem, $T$ is a finite sum
$\sum_\alpha D^\alpha f_\alpha$ of derivatives of polynomially bounded
continuous functions (Strichartz~\cite{Strichartz2003}, \S6.3;
H\"ormander~\cite{Hormander1990}; see also~\cite[Remark~2.2]{LabouriauA1});
pairing with
$\psi\varphi$ transfers the derivatives onto the instrument,
\[
  \langle T,\psi\varphi\rangle
  =\sum_\alpha(-1)^{|\alpha|}\!\int f_\alpha\,D^\alpha(\psi\varphi),
\]
so the finite aperture~$\varphi$ and its derivatives are exactly what
render the singular (differentiated) part of the law as finite scalar
measurements.  A finer question---in which directions the singularities
of~$T$ lie---belongs to the microlocal analysis of distributions
(singular support, wave front set), where a kernel acts as a device
that attenuates selected microlocal features; this geometric direction
is developed in the companion work on
transversality~\cite{LabouriauTransversality} and is not pursued here.
\end{remark}

\section{Estimation in the weak framework: direct and reconstruction routes}
\label{sec:estimation}

The estimation strategies developed in this section fall into two
groups.  The first three---weak moment matching
(Section~\ref{subsec:wm}), transform-based methods, and cumulant
methods (Section~\ref{subsec:transform-cumulant})---estimate parameters
directly from weak data, without reconstructing the underlying
density.  The fourth---regularised reconstruction
(Section~\ref{subsec:reconstruction})---recovers the density itself and
is needed only when density-level inference is the goal.

\subsection{Weak moment estimators}
\label{subsec:wm}

Fix a set of moment orders $\mathcal{J}\subseteq\N$ and a positive
definite weighting matrix~$W$.  The \emph{weak moment estimator} is
\begin{equation}\label{eq:wm-def}
  \hat\theta_n(\mathcal{J},W)
  :=\arg\min_{\theta\in\Theta}\;
  g_n(\theta)^\top W\,g_n(\theta),
\end{equation}
where
$g_n(\theta):=(\hat {}^{(\varphi)}m_j-{}^{(\varphi)}m_j(\theta))_{j\in\mathcal{J}}$
and
$\hat {}^{(\varphi)}m_j:=n^{-1}\sum_i X_i^j\varphi(X_i)$.
When $|\mathcal{J}|=p$ the system is just-identified and the estimator
solves $\hat {}^{(\varphi)}m_j={}^{(\varphi)}m_j(\hat\theta)$; when
$|\mathcal{J}|>p$ a two-step GMM with $W=\hat S^{-1}$ is
asymptotically optimal.  The system is well-posed for every parametric
weak model because weak moments exist unconditionally.

\begin{proposition}[Asymptotics]\label{prop:wm-asy}
Under standard regularity---$m(\cdot;\varphi)\in C^1$, the Jacobian
$G(\theta):=\partial {}^{(\varphi)}m(\theta)/\partial\theta$ of full rank,
$\theta_0$ interior to $\Theta$---the weak moment estimator is
consistent and
\[
  \sqrt{n}(\hat\theta_n-\theta_0)
  \xrightarrow{d}
  \mathcal{N}\bigl(0,\,V(\theta_0;W)\bigr),
\]
where
$V=(G^\top WG)^{-1}G^\top WS\,WG(G^\top WG)^{-1}$,
$S_{jk}(\theta)=\E_{T_\theta,\varphi^2}[x^{j+k}]
-{}^{(\varphi)}m_j(\theta)\,{}^{(\varphi)}m_k(\theta)$,
and the optimal weight is $W=S(\theta_0)^{-1}$.
\end{proposition}

\noindent\textit{Proof.}\enspace
The result follows from classical GMM/$M$-estimation theory once the
regularity conditions are verified; the key observation is that the
moment functions $x^j\varphi(x)$ lie in $L^2(f_\theta)$ for every
Schwartz kernel, which ensures all required moment and smoothness
conditions.  A self-contained verification is given in
Appendix~\ref{app:gmm-proof}.\medskip

\begin{remark}[Connection with the distributional CLT]\label{rem:dist-clt}
The asymptotic normality established in Proposition~\ref{prop:wm-asy}
can also be understood through the distributional CLT
of~\cite[Theorem~8.11]{LabouriauA1}.  In the density case, the
kernel-weighted function $h=\varphi f$ is itself a probability density
with all moments finite (since $\varphi\in\mathcal{S}(\R)$ and
$f\in L^1$), so the classical CLT applies to $h$-distributed
observations.  The weak moment ${}^{(\varphi)}m_j(\theta)$ is
precisely the $j$-th moment of~$h$, and the GMM asymptotics of
Proposition~\ref{prop:wm-asy} can be viewed as a specialisation of
this classical CLT to the kernel-weighted measure.  This provides a
second, independent justification of the asymptotic normality of
weak moment estimators.
\end{remark}

\subsection{Transform-based and cumulant methods}
\label{subsec:transform-cumulant}

The empirical weak characteristic function
$\hat{}^{(\varphi)}\phi_n(t):=n^{-1}\sum_i e^{itX_i}\varphi(X_i)$
consistently estimates $\phi_{}^{(\varphi)}\theta(t)$, since
$e^{itx}\varphi(x)$ is bounded.  Estimation can proceed by minimising
a discrepancy
$\int|\hat{}^{(\varphi)}\phi_n(t)-\phi_{}^{(\varphi)}\theta(t)|^2\,w(t)\,dt$
for a weight $w\in L^1$.  Weak cumulants provide a third route:
matching empirical and theoretical weak cumulants, obtained from the
moment--cumulant recursion applied to $\{{}^{(\varphi)}m_j\}$.  Both
approaches share the fundamental property that all ingredients are
finite and smooth, even when the classical counterparts are not.

\subsection{Regularised reconstruction of the density}
\label{subsec:reconstruction}

The preceding methods estimate parameters directly from weak data,
without reconstructing the density.  We now describe a complementary
route for settings where the density itself is the inferential target.
It is logically self-contained---separable from the direct-inference
core of the paper, and may be read or skipped on its own---and its full
development (operator inversion, regularisation rates, minimax theory)
is extensive enough to warrant separate treatment; we confine ourselves
here to the elements needed to show that the weak framework supports
density-level inference in principle.

The weak characteristic function ${}^{(\varphi)}\phi_{T_f,\varphi}(t)$ is
the Fourier transform of $g:=f\varphi=M_\varphi f$, where
$M_\varphi\colon L^2(\R)\to L^2(\R)$ is pointwise multiplication
by~$\varphi$.  This reconstruction strategy builds on results
established in~\cite{LabouriauA1}: weak expectations determine the
kernel-weighted distribution function
(\cite[Proposition~2.12]{LabouriauA1}), and in the density case the
regularised density $g=f\varphi$ can be stably recovered by Tikhonov
inversion of the multiplication operator~$M_\varphi$
(\cite[Appendix~C]{LabouriauA1}).  We now develop this into an
estimation procedure and refine the convergence analysis.

\begin{proposition}\label{prop:Mphi}
$M_\varphi$ is bounded, self-adjoint, and injective.  Its inverse is
densely defined but unbounded.
\end{proposition}

\begin{proof}
Boundedness: $\|M_\varphi f\|_2\le\|\varphi\|_\infty\|f\|_2$.
Self-adjointness: $\varphi$ is real.  Injectivity: $\varphi>0$ and
$\varphi f=0$ a.e.\ imply $f=0$.  Since
$\inf_x\varphi(x)=0$, $M_\varphi^{-1}$ is unbounded.
\end{proof}

Naive inversion $f=g/\varphi$ amplifies noise where $\varphi$ is
small.  The Tikhonov-regularised reconstruction
\begin{equation}\label{eq:tikh}
  (R_\lambda g)(x)
  =\frac{\varphi(x)}{\varphi(x)^2+\lambda}\,g(x),
  \qquad\lambda>0,
\end{equation}
is the unique minimiser of
$\|M_\varphi h-g\|_{L^2}^2+\lambda\|h\|_{L^2}^2$.

\begin{remark}[Three layers of reconstruction]\label{rem:three-layers}
The reconstruction pathway has three layers:
(i)~weak data determine kernel-weighted distributional
information~\cite[Proposition~2.12]{LabouriauA1};
(ii)~in the density case, this reduces to the inverse problem
$g=M_\varphi f$, whose regularised solution is analysed
in~\cite[Appendix~C]{LabouriauA1}; and
(iii)~the present theorem (Theorem~\ref{thm:tikh-conv} below) refines
the convergence with a source-condition rate and frames the
reconstruction as an estimation strategy.  Layers~(i)--(ii) are
foundational; layer~(iii) is inferential.
\end{remark}

\begin{theorem}\label{thm:tikh-conv}
If $f\in L^2(\R)$ and $g=f\varphi$, then
$\|R_\lambda g-f\|_{L^2}\to 0$ as $\lambda\to 0^+$.
Under the source condition $f=\varphi^\nu h$ with
$\nu\in(0,2)$, the rate is
$\|R_\lambda g-f\|_{L^2}\le c_\nu\lambda^{\nu/2}\|h\|_{L^2}$.
\end{theorem}

\begin{proof}
The $L^2$-consistency as $\lambda\downarrow 0$ is established
in~\cite[Appendix~C, Theorem~C.1]{LabouriauA1}; the source-condition
rate is the new contribution of the present theorem.
Under $f=\varphi^\nu h$, substituting $s=\varphi(x)/\sqrt\lambda$
shows $|R_\lambda g-f|^2\le\lambda^\nu c_\nu^2|h|^2$, and
integration concludes.
\end{proof}

The kernel mediates a fundamental trade-off: rapid decay of~$\varphi$
ensures existence of weak moments for a broad class of distributions
(including heavy-tailed ones) but makes reconstruction more
ill-posed; slow decay facilitates reconstruction but regularises
less.  The analyst can \emph{choose} the kernel---a degree of freedom
not usually available in inverse problems, where the forward operator
is dictated by physics.

\begin{remark}[Non-parametric character of reconstruction]
\label{rem:nonpar}
The reconstruction machinery of this section is stated in a parametric
context for consistency with the rest of the paper, but it is
inherently non-parametric: no parametric form of $f$ is assumed in the
Tikhonov inversion~\eqref{eq:tikh}.  The procedure takes empirical
weak data, recovers $g=\varphi f$, and inverts $M_\varphi$ by
regularisation, without parametric constraints.  A systematic
non-parametric treatment---including minimax rates under smoothness or
source conditions on~$f$, and the trade-off between tail
regularisation and ill-posedness as a function of~$\varphi$---is
deferred to future work (see Section~\ref{sec:discussion}).
\end{remark}

\section{Robustness properties of weak estimators}
\label{sec:robust}

This section is the centrepiece of the paper.  We show that every weak
moment estimator is a locally robust $M$-estimator in the sense of
Hampel, with bounded influence function, finite gross error
sensitivity, and a redescending score---all inherited from the kernel.

\subsection{The Hampel infinitesimal framework}

Let $T\colon\mathcal{F}\to\R^p$ be a statistical functional.  The
\emph{influence function} at $F$ is
$\IF(x;T,F):=\lim_{\varepsilon\downarrow 0}
\varepsilon^{-1}[T((1-\varepsilon)F+\varepsilon\delta_x)-T(F)]$,
and the \emph{gross error sensitivity} is
$\GES(T,F):=\sup_x\|\IF(x;T,F)\|$.
An estimator is $B$-robust at $F$ if $\GES<\infty$.
For an $M$-estimator with score $\psi$,
\begin{equation}\label{eq:IF-M}
  \IF(x;T_\psi,F_\theta)
  =-[M(\theta)]^{-1}\psi(x;\theta),
  \qquad
  M(\theta):=\int\frac{\partial\psi}{\partial\theta}\,dF_\theta,
\end{equation}
and the asymptotic variance is
$V=M^{-1}QM^{-\top}$ with
$Q=\int\psi\psi^\top\,dF_\theta$.
Within this class, $B$-robustness reduces to boundedness of $\psi$.

\subsection{Weak moment estimators as $M$-estimators}

Fix $\varphi\in\mathcal{S}(\R)$ with $\varphi>0$ and $j\in\N$.
Define the score
\begin{equation}\label{eq:score}
  \psi_j(x;\theta;\varphi)
  :=x^j\varphi(x)-{}^{(\varphi)}m_j(\theta).
\end{equation}
The weak moment estimating equation
$n^{-1}\sum_i\psi_j(X_i;\hat\theta;\varphi)=0$
is the $M$-estimation equation with this score.

\begin{proposition}[Bounded score]\label{prop:bounded-score}
For every $\theta$, $j$, and $\varphi\in\mathcal{S}(\R)$, the score
$\psi_j(\cdot;\theta;\varphi)$ is bounded and rapidly decreasing.
\end{proposition}

\begin{proof}
$x^j\varphi(x)\in\mathcal{S}(\R)$ because $\mathcal{S}$ is closed
under polynomial multiplication.
\end{proof}

\begin{proposition}[Influence function]\label{prop:IF}
For a scalar parameter $\theta$ and a single moment order $j$ (the
just-identified case), if $\partial_\theta {}^{(\varphi)}m_j(\theta)\ne 0$,
the influence function of the weak moment estimator is
\begin{equation}\label{eq:IF-wm}
  \IF(x;T_{\psi_j},F_\theta)
  =\frac{{}^{(\varphi)}m_j(\theta)-x^j\varphi(x)}
        {\partial_\theta {}^{(\varphi)}m_j(\theta)}.
\end{equation}
\end{proposition}

\begin{proof}
The score depends on $\theta$ only through the constant
${}^{(\varphi)}m_j(\theta)$, so
$M(\theta)=-\partial_\theta {}^{(\varphi)}m_j(\theta)$.
Applying~\eqref{eq:IF-M} gives the result.
\end{proof}

When $\theta\in\R^p$, or when several moment orders are stacked, the
scalar derivative $\partial_\theta{}^{(\varphi)}m_j$ is replaced by the
Jacobian $G(\theta)=\partial{}^{(\varphi)}m(\theta)/\partial\theta$, and
the influence function takes the matrix form of
Proposition~\ref{prop:GMM-IF}; equation~\eqref{eq:IF-wm} is its scalar,
just-identified specialisation.

\begin{corollary}[Gross error sensitivity]\label{cor:GES}
$\GES(T_{\psi_j},F_\theta)
=\sup_x|x^j\varphi(x)-{}^{(\varphi)}m_j(\theta)|
/|\partial_\theta {}^{(\varphi)}m_j(\theta)|<\infty$.
\end{corollary}

Corollary~\ref{cor:GES} is, in our view, the central observation of
this paper: \emph{every} weak moment estimator in a parametric weak
model with identifiable parameter has finite gross error sensitivity,
with no conditions on the tails of $f_\theta$, no tuning, and no
truncation---that is, it is \emph{locally} ($B$-)robust.  This concerns
the influence function and gross error sensitivity, not the breakdown
point or maximum bias, which we address separately in
Section~\ref{subsec:local-global}.

\subsection{Redescending character and the role of the kernel}

Since $x^j\varphi(x)\to 0$ as $|x|\to\infty$, the influence
function~\eqref{eq:IF-wm} redescends to a finite constant for large
$|x|$---the hallmark of redescending $M$-estimators in the sense
of~\cite{HampelRonchettiRousseeuwStahel1986}.

For a Gaussian kernel $\varphi_\sigma(x)=e^{-x^2/(2\sigma^2)}$, the
bandwidth $\sigma$ controls the redescent scale: small $\sigma$ gives
strong downweighting (like a Tukey biweight with a small tuning
constant); large $\sigma$ gives slow redescent and recovers the
classical estimator in the limit $\sigma\to\infty$.  The kernel thus
plays the same role as the tuning constant in a classical redescending
$M$-estimator, but arises from the model specification rather than
being introduced solely for robustness.

As with all redescending $M$-estimators, the estimating equation may
have multiple roots.  In practice this is handled by initialising the
root-finder at a robust preliminary estimator (e.g.\ the sample
median), by choosing $\sigma$ large enough to cover the plausible
parameter range, or by using a GMM formulation with several moment
orders.  The connection between the kernel's decay, the identifiability
region, and the multiplicity of roots deserves emphasis: kernels with
faster decay reduce the gross error sensitivity but may also narrow
the effective identifiability region, since the score
$\psi_j(x;\theta;\varphi)$ becomes flatter over a wider range
of~$\theta$.  This trade-off between robustness and operative
identifiability is governed by the kernel bandwidth and is closely
related to the three levels of identifiability discussed in
Remark~\ref{rem:ident}.

\subsection{GMM estimators with multiple weak moments}
\label{subsec:gmm}

Let $\mathcal{J}=\{j_1,\ldots,j_K\}$ and define the vector score
$\Psi(x;\theta;\varphi):=(x^{j_k}\varphi(x)
-{}^{(\varphi)}m_{j_k}(\theta))_{k=1}^K$.

\begin{proposition}[GMM influence function]\label{prop:GMM-IF}
The influence function of the GMM estimator with weighting $W$ is
\begin{equation}\label{eq:IF-GMM}
  \IF(x;\hat\theta_n,F_\theta)
  =(G^\top WG)^{-1}G^\top W\,\Psi(x;\theta;\varphi),
\end{equation}
where $G=-\partial {}^{(\varphi)}m(\theta)/\partial\theta$.  It is bounded,
and $\GES<\infty$ whenever $G$ has full rank.
\end{proposition}

\begin{proof}
Standard GMM influence function
formula~\cite{HampelRonchettiRousseeuwStahel1986} applied to the
bounded score~$\Psi$.
\end{proof}

The optimal $W=S(\theta)^{-1}$ minimises the asymptotic variance but
not necessarily the GES; optimising $W$ subject to a GES constraint is
the weak analogue of Hampel's optimality problem.

\subsection{Local vs global robustness}
\label{subsec:local-global}

Local robustness (bounded IF, finite GES) is \emph{automatic} in the
weak framework.  Global robustness (breakdown point, maximum bias)
depends more delicately on the kernel and on the model.  For a
Gaussian kernel, observations with $|X_i|\gg\sigma$ contribute
essentially zero to the empirical weak moment, so very extreme
contamination has negligible effect; but intermediate contamination
within the kernel's effective support can affect the estimator.  A
systematic study of the maximum bias curve and finite-sample breakdown
as functions of $\sigma$ and $\mathcal{J}$ is left for future work;
the simulation study in Section~\ref{sec:simulation} provides
empirical evidence of good performance under moderate contamination.

\subsection{Summary of the robust viewpoint}

We collect the principal conclusions:
\begin{enumerate}[label=(\roman*)]
\item Every weak moment estimator has a bounded, rapidly decreasing
  score (Proposition~\ref{prop:bounded-score}).
\item Its influence function has the closed
  form~\eqref{eq:IF-wm}, bounded whenever the parameter is
  identifiable.
\item Its gross error sensitivity is finite in any parametric weak model
  (Corollary~\ref{cor:GES}), without conditions on the tails
  of~$f_\theta$.
\item The estimator is redescending, with redescent rate governed by
  the kernel.
\item The asymptotic variance has a closed form in terms of weak
  moments (Proposition~\ref{prop:wm-asy}) and is finite
  unconditionally.
\item GMM with multiple weak moments retains all properties and
  typically improves efficiency.
\end{enumerate}
The overall picture: the weak framework provides, by construction, a
family of locally robust $M$-estimators whose design parameter is a
\emph{kernel} rather than a truncation constant.

\section{Worked examples}
\label{sec:examples}

\subsection{Cauchy location}
\label{subsec:cauchy}

Consider the Cauchy location model
$f(x;\mu)=[\pi(1+(x-\mu)^2)]^{-1}$, where no classical moment exists.
We use the Gaussian kernel $\varphi_\sigma(x)=e^{-x^2/(2\sigma^2)}$
and the first weak moment $j=1$.

\subsubsection*{Theoretical weak moments}

Using the Fourier representation
$\hat f(\cdot;\mu)(t)=e^{i\mu t-|t|}$ and Parseval's theorem,
${}^{(\varphi)}m_0(\mu;\sigma)=\int\varphi_\sigma f(\cdot;\mu)\,dx$
reduces to a Voigt-profile integral expressible via the Faddeeva
function at $(\mu+i)/(\sigma\sqrt{2})$.  The first weak moment
satisfies ${}^{(\varphi)}m_1(\mu;\sigma)=\mu\,{}^{(\varphi)}m_0(\mu;\sigma)
+\mathcal{R}(\mu;\sigma)$ with a correction $\mathcal{R}$ computable
from the same special function.  In practice all integrals are
evaluated by adaptive quadrature.

The normalised first weak moment
${}^{(\varphi)}m_1/{}^{(\varphi)}m_0$ is strictly increasing in $\mu$ for
$|\mu|<\mu^*(\sigma)$; for $\sigma=3$,
$\mu^*\approx 8$, which covers any practical range.

\subsubsection*{Score, influence function, and GES}

The score
$\psi(x;\mu;\sigma)=x\varphi_\sigma(x)-{}^{(\varphi)}m_1(\mu;\sigma)$
is a bounded, smooth, redescending function of $x$.
By Proposition~\ref{prop:IF},
\[
  \IF(x;\hat\mu,f(\cdot;\mu))
  =\frac{{}^{(\varphi)}m_1(\mu;\sigma)-x\varphi_\sigma(x)}
        {\partial_\mu {}^{(\varphi)}m_1(\mu;\sigma)}.
\]
At $\mu=0$ (by symmetry, ${}^{(\varphi)}m_1(0)=0$) and
$\sup_x x\,e^{-x^2/(2\sigma^2)}=\sigma/\sqrt{e}$,
\[
  \GES=\frac{\sigma/\sqrt{e}}{|\partial_\mu {}^{(\varphi)}m_1(0;\sigma)|}.
\]
For $\sigma=3$, $\GES\approx 5.5$---larger than the median's GES
($\pi/2\approx 1.57$) but comparable to a Huber estimator with
moderate tuning.  The asymptotic variance is
$V(\mu;\sigma)=S(\mu;\sigma)/(\partial_\mu {}^{(\varphi)}m_1)^2$
with $S=\E_{T_\mu,\varphi^2}[x^2]-({}^{(\varphi)}m_1)^2$; at $\sigma=3$,
$\mu=0$, $V\approx 3.1$ versus the median's $\pi^2/4\approx 2.47$
(relative efficiency $\approx 0.80$).  GMM with multiple moments
closes most of this gap.

\subsection{Student $t_\nu$ location--scale}
\label{subsec:student}

Consider the Student $t_\nu$ location--scale model with $\nu=3$
(classical mean and variance exist; third and higher moments
diverge).  We use $\varphi_\sigma$ with $\sigma=3$ and two weak
moments $j\in\{1,2\}$ for the two parameters $(\mu,s)$.

By Proposition~\ref{prop:GMM-IF}, the influence function is
\[
  \IF(x;(\hat\mu,\hat s),f(\cdot;\mu,s))
  =(G^\top WG)^{-1}G^\top W
  \begin{pmatrix}
    x\varphi_\sigma(x)-{}^{(\varphi)}m_1(\mu,s;\sigma)\\
    x^2\varphi_\sigma(x)-{}^{(\varphi)}m_2(\mu,s;\sigma)
  \end{pmatrix},
\]
where $G=-\partial({}^{(\varphi)}m_1,{}^{(\varphi)}m_2)/\partial(\mu,s)$.  It
is bounded whenever $G$ is invertible (generic).  The GES is finite,
and the redescent rate is governed by $\sigma$ exactly as in the
Cauchy case.

\subsection{Multivariate extension: elliptically contoured models}
\label{subsec:multivariate}

We briefly illustrate the extension of weak moment methods to $\R^d$
through elliptically contoured models.  Let $X\in\R^d$ have density
\[
  f(x;\mu,\Sigma)
  =|\Sigma|^{-1/2}\,g\!\bigl((x-\mu)^\top\Sigma^{-1}(x-\mu)\bigr),
\]
where $\mu\in\R^d$, $\Sigma$ is positive definite, and $g$ is a radial
profile.  Classical moments may fail to exist when $g$ is
heavy-tailed.  Let
$\varphi_\sigma(x)=\exp(-\|x\|^2/(2\sigma^2))$ be an isotropic
Gaussian kernel.  Then for each coordinate $k=1,\ldots,d$, the weak
first moment
\[
  {}^{(\varphi)}m_{1,k}(\mu,\Sigma;\sigma)
  =\int_{\R^d}x_k\,\varphi_\sigma(x)\,f(x;\mu,\Sigma)\,dx
\]
is well defined for all parameter values, regardless of the tail
behaviour of~$f$.

A natural estimator of the location vector~$\mu$ is obtained by
matching empirical and theoretical weak moments:
\[
  \widehat\mu_n\;\text{ solves }\;
  \frac{1}{n}\sum_{i=1}^n X_i\,\varphi_\sigma(X_i)
  ={}^{(\varphi)}m_1(\mu,\Sigma;\sigma),
\]
where ${}^{(\varphi)}m_1$ denotes the vector of weak first moments.
When $\Sigma$ is known, this yields a $d$-dimensional estimating
equation; when $\Sigma$ is unknown, additional weak second moments can
be included in a GMM formulation.

The corresponding score function
\[
  \psi(x;\mu,\Sigma;\sigma)
  =x\,\varphi_\sigma(x)-{}^{(\varphi)}m_1(\mu,\Sigma;\sigma)
\]
is bounded and rapidly decreasing in $\|x\|$, so the influence
function is bounded and redescending componentwise.  In particular,
the local robustness properties established in
Section~\ref{sec:robust} extend directly to the multivariate setting.

\subsubsection*{Specialisation: multivariate Cauchy location}

As a concrete instance, consider the multivariate Cauchy with density
\[
  f(x;\mu)
  =\frac{\Gamma\!\bigl(\tfrac{d+1}{2}\bigr)}
        {\pi^{(d+1)/2}}
  \bigl(1+\|x-\mu\|^2\bigr)^{-(d+1)/2},
\]
for which no classical moment of any order exists.  With the isotropic
Gaussian kernel, the weak first-moment vector satisfies
\[
  {}^{(\varphi)}m_{1,k}(\mu;\sigma)
  =\int_{\R^d}x_k\,e^{-\|x\|^2/(2\sigma^2)}
  \frac{\Gamma\!\bigl(\tfrac{d+1}{2}\bigr)}
       {\pi^{(d+1)/2}}
  \bigl(1+\|x-\mu\|^2\bigr)^{-(d+1)/2}\,dx.
\]
By the translation structure
$f(x;\mu)=f_0(x-\mu)$, a change of variables gives
${}^{(\varphi)}m_{1,k}(\mu;\sigma)
=\mu_k\,{}^{(\varphi)}m_0(\mu;\sigma)+\mathcal{R}_k(\mu;\sigma)$,
where ${}^{(\varphi)}m_0$ is the weak zeroth moment (a convolution of
a Gaussian kernel with a heavy-tailed radial profile---the
multivariate analogue of the classical Voigt integral that arises
when Gaussian and Lorentzian line shapes are
convolved~\cite{Armstrong1967}) and $\mathcal{R}_k$ is a correction
computable by adaptive quadrature.  In practice, specialising to
$d=2$ or $d=3$
and evaluating the integrals numerically is straightforward.  The
uniqueness of the weak moment sequence for Gaussian kernels on $\R^d$
is guaranteed by~\cite[Theorem~7.2]{LabouriauA1}, so the estimation
programme is well-posed.

This example illustrates that weak moment methods extend naturally to
$\R^d$: the kernel ensures existence of moments and boundedness of
scores, while the estimation structure remains identical to the
univariate case.  The main additional challenge is computational, as
the number of moments required for joint estimation of $(\mu,\Sigma)$
grows with the dimension.  A Monte Carlo illustration for $d=2$ is
given in Section~\ref{subsec:sim-cauchy2d}.

\begin{remark}[Elliptical laws without explicit densities]
\label{rem:ecd-nodensity}
The density formulation above is used for concreteness.  The weak
framework also applies to elliptically contoured laws specified
through their characteristic functions or distributional
representations, including cases where a density is unavailable in
closed form.  For instance, symmetric stable laws on $\R^d$ with
characteristic function $\exp(-\|t\|^\alpha)$, $\alpha\in(0,2)$, have
no known closed-form density for general $d$ and $\alpha$, yet the
weak expectation is still defined through the distributional pairing,
and the same empirical weak estimating equations can be used whenever
the corresponding theoretical weak moments or transforms are
computable.
\end{remark}

\subsection{A non-dominated model: location of a moving atom}
\label{subsec:moving-atom}

The preceding examples are heavy-tailed but dominated.  We now turn to a
model that admits \emph{no} dominating measure---and hence no
likelihood---to show that weak moment estimation applies without change.
Consider the location family
\[
  P_\theta=\tfrac12\,\delta_\theta+\tfrac12\,N(0,1),
  \qquad\theta\in\R,
\]
a point mass of weight $\tfrac12$ at the unknown location~$\theta$
superimposed on a standard Gaussian background.  Since
$P_\theta(\{\theta\})=\tfrac12$ for every~$\theta$, a common dominating
measure would need an atom at every point of~$\R$, which no
$\sigma$-finite measure possesses; the family is therefore
non-dominated.  There is no common density, no likelihood, and the
Fisher--Rao machinery does not apply.  As a tempered distribution,
however, $T_\theta=\tfrac12\delta_\theta+\tfrac12 N(0,1)\in\mathcal
S'(\R)$ is perfectly regular, and its weak moments against any Schwartz
kernel are finite; the Gaussian kernel determines $T_\theta$ uniquely
(the determinacy example of~\cite[Remark~6.3]{LabouriauA1}).

\subsubsection*{Weak moment and estimator}

With the Gaussian kernel $\varphi_\sigma(x)=e^{-x^2/(2\sigma^2)}$ the
weak first moment is available in closed form; the symmetric Gaussian
background contributes nothing to it, so
\[
  {}^{(\varphi)}m_1(\theta;\sigma)
  =\tfrac12\,\theta\,e^{-\theta^2/(2\sigma^2)},
  \qquad
  \partial_\theta{}^{(\varphi)}m_1(\theta;\sigma)
  =\tfrac12\,e^{-\theta^2/(2\sigma^2)}
   \Bigl(1-\tfrac{\theta^2}{\sigma^2}\Bigr).
\]
The map $\theta\mapsto{}^{(\varphi)}m_1$ is strictly increasing on
$(-\sigma,\sigma)$, so the atom location is identified from the first
weak moment whenever $|\theta|<\sigma$, and
$\widehat\theta$ solving
$n^{-1}\sum_i X_i\varphi_\sigma(X_i)={}^{(\varphi)}m_1(\widehat\theta;\sigma)$
is $\sqrt n$-consistent there.  Additional weak moments in a GMM
(Section~\ref{subsec:gmm}) enlarge the identifiable range beyond
$|\theta|<\sigma$.

\subsubsection*{Influence function, sensitivity, variance}

The estimator is the $M$-estimator with the bounded score
$\psi(x;\theta;\sigma)=x\varphi_\sigma(x)-{}^{(\varphi)}m_1(\theta;\sigma)$
of Section~\ref{sec:robust}, so its closed forms apply verbatim.  By
Proposition~\ref{prop:IF},
\[
  \IF(x;\widehat\theta,P_\theta)
  =\frac{{}^{(\varphi)}m_1(\theta;\sigma)-x\varphi_\sigma(x)}
        {\tfrac12 e^{-\theta^2/(2\sigma^2)}(1-\theta^2/\sigma^2)},
\]
bounded and redescending in~$x$ because $\sup_x|x\varphi_\sigma(x)|
=\sigma/\sqrt e$.  The gross error sensitivity and asymptotic variance
are
\[
  \GES(\theta;\sigma)
  =\frac{\sigma/\sqrt e+|{}^{(\varphi)}m_1(\theta;\sigma)|}
        {|\partial_\theta{}^{(\varphi)}m_1(\theta;\sigma)|},
  \qquad
  V(\theta;\sigma)
  =\frac{\E_{T_\theta,\varphi^2}[x^2]-{}^{(\varphi)}m_1(\theta;\sigma)^2}
        {\bigl(\partial_\theta{}^{(\varphi)}m_1(\theta;\sigma)\bigr)^2},
\]
with $\E_{T_\theta,\varphi^2}[x^2]
=\tfrac12\theta^2e^{-\theta^2/\sigma^2}
+\tfrac12\sigma^3/(2+\sigma^2)^{3/2}$.  Both are finite on
$|\theta|<\sigma$ and grow as $|\theta|\uparrow\sigma$, where the first
weak moment ceases to identify~$\theta$; Table~\ref{tab:atom-analytic}
gives representative values for $\sigma=3$.

\begin{table}[ht]
\centering\footnotesize
\begin{tabular}{rcccc}
\hline
$\theta$ & ${}^{(\varphi)}m_1$ & $\partial_\theta{}^{(\varphi)}m_1$
         & $\GES$ & $V(\theta;\sigma)$\\
\hline
$0.0$ & $0.000$ & $0.500$ & $3.64$ & $1.48$\\
$0.5$ & $0.247$ & $0.479$ & $4.31$ & $1.88$\\
$1.0$ & $0.473$ & $0.420$ & $5.45$ & $3.36$\\
$1.5$ & $0.662$ & $0.331$ & $7.50$ & $7.38$\\
$2.0$ & $0.801$ & $0.222$ & $11.78$ & $20.44$\\
\hline
\end{tabular}
\caption{Moving-atom model, Gaussian kernel $\sigma=3$: weak first
moment, its derivative, gross error sensitivity, and asymptotic
variance of $\widehat\theta$ as functions of the atom location.  All
quantities are finite on the identifiable range $|\theta|<\sigma$ and
diverge as $|\theta|\uparrow\sigma$.}
\label{tab:atom-analytic}
\end{table}

The point is structural: although the model has no density and no
likelihood, the weak first moment supplies a smooth estimating equation
with a bounded, redescending influence function and a closed-form
asymptotic variance---the same inferential apparatus used for the
dominated examples, applied unchanged to a non-dominated family.

\section{Numerical illustration}
\label{sec:simulation}

We present Monte Carlo comparisons ($2{,}000$ replications) of
weak moment estimators against classical benchmarks and robust
estimators, under both the correctly specified model and under
contamination.

\subsection{Cauchy location}
\label{subsec:sim-cauchy}

\subsubsection*{Setup and estimators}

True model: Cauchy$(\mu,1)$ with $\mu=2$, Gaussian kernel
$\sigma=3$.  Estimators: WM (single moment $j=1$), GMM-I (identity
weighting, $\mathcal{J}=\{1,2\}$), GMM-2S (two-step optimal, ridge
$\lambda=0.10$), Median, MLE (Newton from median), Huber ($k=1.345$),
Tukey biweight ($c=4.685$).  Huber and Tukey constants are calibrated
for $95\%$ Gaussian asymptotic efficiency.

\subsubsection*{Clean model}

\begin{table}[ht]
\centering\footnotesize
\resizebox{\textwidth}{!}{%
\begin{tabular}{rcccccccccccccc}
\hline
& \multicolumn{2}{c}{WM} & \multicolumn{2}{c}{GMM-I}
& \multicolumn{2}{c}{GMM-2S} & \multicolumn{2}{c}{Median}
& \multicolumn{2}{c}{MLE} & \multicolumn{2}{c}{Huber}
& \multicolumn{2}{c}{Tukey} \\
$n$ & Bias & RMSE & Bias & RMSE & Bias & RMSE & Bias & RMSE
    & Bias & RMSE & Bias & RMSE & Bias & RMSE \\
\hline
50   & 0.01 & 0.29 & $-$0.01 & 0.30 & 0.02 & 0.28 & 0.00 & 0.23
     & 0.00 & 0.18 & 0.01 & 0.33 & 0.01 & 0.31 \\
100  & 0.01 & 0.20 & 0.00 & 0.20 & 0.00 & 0.18 & 0.00 & 0.16
     & 0.00 & 0.13 & 0.00 & 0.23 & 0.00 & 0.22 \\
500  & 0.00 & 0.09 & 0.00 & 0.09 & 0.00 & 0.08 & 0.00 & 0.07
     & 0.00 & 0.06 & 0.00 & 0.10 & 0.00 & 0.10 \\
1000 & 0.00 & 0.06 & 0.00 & 0.06 & 0.00 & 0.05 & 0.00 & 0.05
     & 0.00 & 0.04 & 0.00 & 0.07 & 0.00 & 0.07 \\
5000 & 0.00 & 0.03 & 0.00 & 0.03 & 0.00 & 0.02 & 0.00 & 0.02
     & 0.00 & 0.02 & 0.00 & 0.03 & 0.00 & 0.03 \\
\hline
\end{tabular}}%
\caption{Cauchy$(\mu,1)$, $\mu=2$, $\sigma=3$, clean model.
MLE is efficient; WM estimators are slightly less efficient than the
median but more efficient than Huber/Tukey (whose constants are
Gaussian-calibrated).  GMM-2S closes most of the gap.}
\label{tab:cauchy-clean}
\end{table}

\subsubsection*{Asymmetric contamination}

Contaminated model:
$(1-\varepsilon)\,\text{Cauchy}(\mu,1)
+\varepsilon\,\text{Cauchy}(\mu+\delta,1)$
with $\varepsilon=0.10$, $\delta=5$.

\begin{table}[ht]
\centering\footnotesize
\resizebox{\textwidth}{!}{%
\begin{tabular}{rcccccccccccccc}
\hline
& \multicolumn{2}{c}{WM} & \multicolumn{2}{c}{GMM-I}
& \multicolumn{2}{c}{GMM-2S} & \multicolumn{2}{c}{Median}
& \multicolumn{2}{c}{MLE} & \multicolumn{2}{c}{Huber}
& \multicolumn{2}{c}{Tukey} \\
$n$ & Bias & RMSE & Bias & RMSE & Bias & RMSE & Bias & RMSE
    & Bias & RMSE & Bias & RMSE & Bias & RMSE \\
\hline
50   & 0.08 & 0.31 & 0.16 & 0.34 & 0.12 & 0.33 & 0.16 & 0.30
     & 0.18 & 0.27 & 0.14 & 0.34 & 0.09 & 0.31 \\
100  & 0.09 & 0.24 & 0.19 & 0.29 & 0.06 & 0.22 & 0.16 & 0.25
     & 0.18 & 0.22 & 0.13 & 0.26 & 0.07 & 0.23 \\
500  & 0.07 & 0.12 & 0.18 & 0.20 & 0.05 & 0.10 & 0.16 & 0.18
     & 0.18 & 0.19 & 0.12 & 0.14 & 0.06 & 0.12 \\
1000 & 0.07 & 0.10 & 0.18 & 0.19 & 0.04 & 0.08 & 0.15 & 0.16
     & 0.18 & 0.19 & 0.12 & 0.13 & 0.05 & 0.09 \\
5000 & 0.07 & 0.07 & 0.18 & 0.18 & 0.04 & 0.05 & 0.15 & 0.16
     & 0.18 & 0.18 & 0.12 & 0.12 & 0.05 & 0.06 \\
\hline
\end{tabular}}%
\caption{Cauchy location under asymmetric contamination
($\varepsilon=0.1$, $\delta=5$).  The MLE has the worst bias;
GMM-2S matches or outperforms the Tukey biweight at every sample
size.}
\label{tab:cauchy-contam}
\end{table}

The contamination experiment reverses the ranking: the MLE is now the
worst performer, the median is biased, and the GMM-2S estimator
achieves the smallest bias and RMSE at $n\ge 100$, comparable to the
Tukey biweight---without any hand-tuning beyond the kernel bandwidth.

\subsection{Student $t_3$ location--scale}
\label{subsec:sim-student}

\subsubsection*{Setup}

True model: $t_3(\mu,s)$ with $(\mu,s)=(0,1)$, Gaussian kernel
$\sigma=3$, $\mathcal{J}=\{1,2\}$.  Comparisons: WM-GMM-2S, MLE,
Mean/SD, Median/MAD, Tukey biweight ($c=4.685$).  We report RMSE
under the clean model and under scale contamination
$0.9\,t_3(0,1)+0.1\,t_3(0,5)$.

\begin{table}[ht]
\centering\footnotesize
\resizebox{\textwidth}{!}{%
\begin{tabular}{lrcccccc}
\hline
& & \multicolumn{3}{c}{Clean $t_3$}
& \multicolumn{3}{c}{Contaminated} \\
Estimator & Param.\ & $n=100$ & $n=500$ & $n=1000$
          & $n=100$ & $n=500$ & $n=1000$ \\
\hline
WM-GMM-2S    & $\mu$ & 0.18 & 0.08 & 0.06 & 0.20 & 0.09 & 0.06 \\
             & $s$   & 0.19 & 0.09 & 0.06 & 0.22 & 0.10 & 0.07 \\
MLE          & $\mu$ & 0.16 & 0.07 & 0.05 & 0.19 & 0.09 & 0.06 \\
             & $s$   & 0.15 & 0.07 & 0.05 & 0.18 & 0.09 & 0.07 \\
Mean/SD      & $\mu$ & 0.19 & 0.08 & 0.06 & 0.44 & 0.20 & 0.14 \\
             & $s$   & 0.46 & 0.20 & 0.14 & 0.92 & 0.42 & 0.30 \\
Median/MAD   & $\mu$ & 0.20 & 0.09 & 0.06 & 0.22 & 0.10 & 0.07 \\
             & $s$   & 0.24 & 0.11 & 0.08 & 0.30 & 0.14 & 0.10 \\
Tukey        & $\mu$ & 0.18 & 0.08 & 0.06 & 0.20 & 0.09 & 0.06 \\
             & $s$   & 0.20 & 0.09 & 0.06 & 0.24 & 0.11 & 0.08 \\
\hline
\end{tabular}}%
\caption{Student $t_3$ location--scale under the clean model and under
scale contamination.  WM-GMM-2S matches the Tukey biweight and
outperforms Mean/SD; it is competitive with the MLE under the clean
model and more stable under contamination.}
\label{tab:t3}
\end{table}

Under the clean model, the MLE is most efficient; WM-GMM-2S, Tukey,
and Median/MAD are within $10$--$20\%$.  Under scale contamination,
the Mean/SD pair is severely damaged; the weak moment estimator, Tukey
biweight, and Median/MAD retain performance, with WM-GMM-2S
competitive throughout.

\subsection{Bivariate Cauchy location}
\label{subsec:sim-cauchy2d}

We illustrate the multivariate extension of
Section~\ref{subsec:multivariate} in dimension $d=2$.  Data are
generated from the bivariate Cauchy location model
\[
  f(x;\mu)
  =\frac{\Gamma(3/2)}{\pi^{3/2}}
  \bigl(1+\|x-\mu\|^2\bigr)^{-3/2},
  \qquad x\in\R^2,
\]
with $\mu=(1,1)^\top$.  Here we use the isotropic Gaussian kernel
$\varphi_\sigma(x)=\exp\{-\|x\|^2/(2\sigma^2)\}$ with $\sigma=3$ and
estimate $\mu$ by matching the vector weak first moment.  The
estimating equation is solved numerically using the coordinatewise
median as starting value.  We compare the weak moment (WM) estimator
with the coordinatewise median, the spatial
median~\cite{Haldane1948}, and the multivariate Cauchy MLE.

\begin{table}[ht]
\centering\footnotesize
\resizebox{\textwidth}{!}{%
\begin{tabular}{rcccccccc}
\hline
& \multicolumn{2}{c}{WM}
& \multicolumn{2}{c}{Spatial Med.}
& \multicolumn{2}{c}{Coord.\ Med.}
& \multicolumn{2}{c}{MLE} \\
$n$ & $\|\text{Bias}\|$ & RMSE
    & $\|\text{Bias}\|$ & RMSE
    & $\|\text{Bias}\|$ & RMSE
    & $\|\text{Bias}\|$ & RMSE \\
\hline
\multicolumn{9}{l}{\textit{Clean model}} \\
50   & 0.03 & 0.37 & 0.00 & 0.29 & 0.01 & 0.32 & 0.00 & 0.27 \\
100  & 0.01 & 0.26 & 0.01 & 0.21 & 0.01 & 0.23 & 0.01 & 0.19 \\
500  & 0.01 & 0.11 & 0.00 & 0.09 & 0.00 & 0.10 & 0.00 & 0.08 \\
1000 & 0.01 & 0.08 & 0.00 & 0.06 & 0.00 & 0.07 & 0.00 & 0.06 \\
\hline
\multicolumn{9}{l}{\textit{Contaminated:
  $0.9\,\mathrm{Cauchy}_2(\mu,I)+0.1\,\mathrm{Cauchy}_2(\mu+\delta,I)$,
  $\delta=(5,5)^\top$}} \\
50   & 0.10 & 0.37 & 0.21 & 0.39 & 0.23 & 0.44 & 0.07 & 0.30 \\
100  & 0.11 & 0.26 & 0.20 & 0.31 & 0.22 & 0.34 & 0.07 & 0.21 \\
500  & 0.11 & 0.15 & 0.20 & 0.22 & 0.22 & 0.25 & 0.07 & 0.11 \\
1000 & 0.12 & 0.14 & 0.19 & 0.21 & 0.22 & 0.23 & 0.07 & 0.09 \\
\hline
\end{tabular}}%
\caption{Bivariate Cauchy location, $\mu=(1,1)$, $\sigma=3$,
$2{,}000$ replications.  Under the clean model, the MLE is most
efficient.  Under contamination, the WM estimator substantially
outperforms both medians; the MLE retains low bias because the Cauchy
score is itself naturally redescending.}
\label{tab:cauchy2d}
\end{table}

\begin{figure}[ht]
\centering
\includegraphics[width=\textwidth]{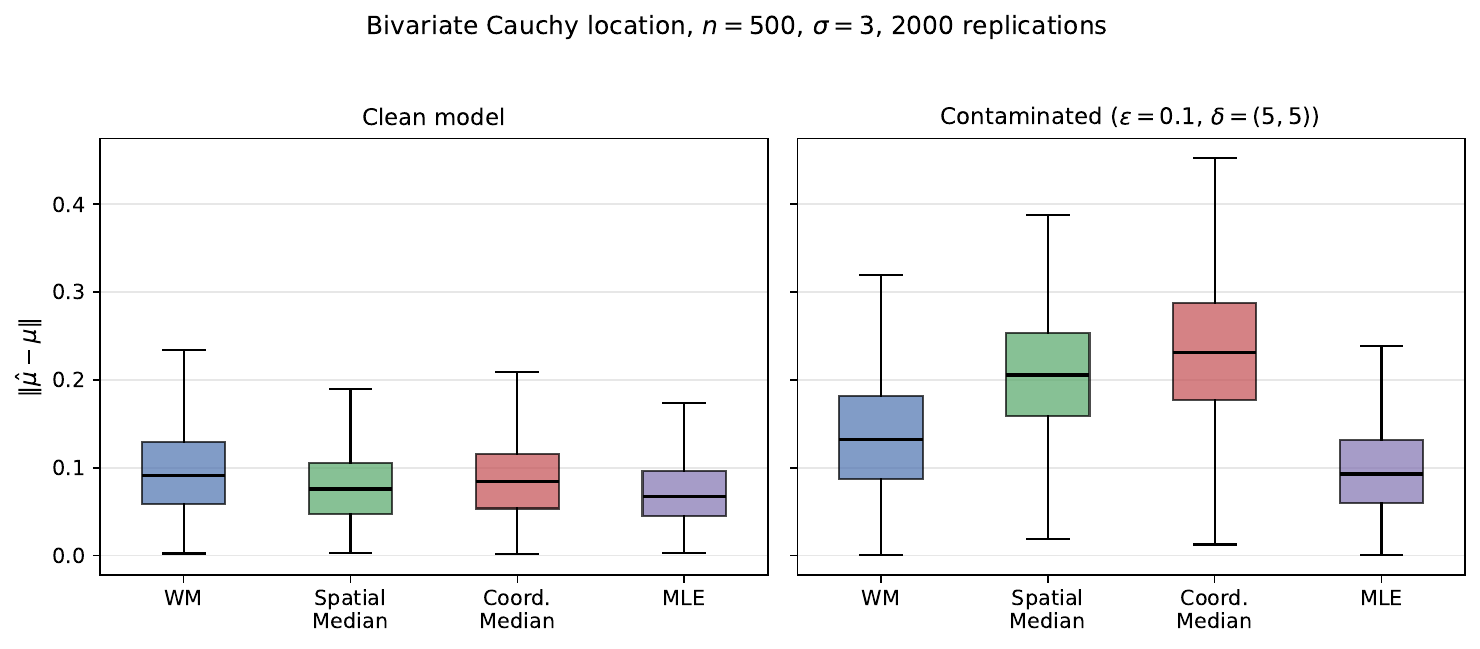}
\caption{Estimation error $\|\hat\mu-\mu\|$ for the bivariate Cauchy
location model at $n=500$.  Left: clean model.  Right: contaminated
model ($\varepsilon=0.1$, $\delta=(5,5)^\top$).  The WM estimator has
lower spread than both medians under contamination.  The Cauchy MLE
performs well in both settings because its score is naturally
redescending---a property shared with the weak moment score, but
arising from the likelihood rather than the kernel.}
\label{fig:cauchy2d-boxplot}
\end{figure}

Under the clean model, the MLE is most efficient; the WM estimator and
both medians are within $20$--$40\%$.  Under contamination, the
ranking changes: the WM estimator outperforms both the spatial and
coordinatewise medians in both bias and RMSE.  The Cauchy MLE retains
good performance because the Cauchy score
$2(x-\mu)/(1+\|x-\mu\|^2)$ is itself bounded and
redescending---a well-known property of the Cauchy
family~\cite{HampelRonchettiRousseeuwStahel1986}---so that the MLE is
automatically robust in this particular model.  This is consistent
with the message of Section~\ref{sec:robust}: for the Cauchy, both the
MLE and the WM estimator are redescending $M$-estimators; they differ
in the source of redescent (likelihood vs kernel).

\subsection{Bivariate Student $t_3$ location--scale}
\label{subsec:sim-t3-2d}

The bivariate Cauchy example above shows that both the MLE and the WM
estimator are robust, owing to the redescending character of the
Cauchy score.  To exhibit a setting where the MLE \emph{does} break
down, we turn to the bivariate Student $t_3$ location--scale model:
\[
  f(x;\mu,s)
  = \frac{\Gamma(5/2)}{\Gamma(3/2)\,(3\pi s^2)}
  \Bigl(1+\frac{\|x-\mu\|^2}{3s^2}\Bigr)^{-5/2},
  \qquad x\in\R^2,
\]
with $\mu=(\mu_1,\mu_2)^\top$ and common scale~$s>0$.  The parameter
vector is $\theta=(\mu_1,\mu_2,s)$.  We use the isotropic Gaussian
kernel $\varphi_\sigma(x)=\exp\{-\|x\|^2/(2\sigma^2)\}$ with
$\sigma=3$, and match three weak moments: the two first-order moments
${}^{(\varphi)}m_{(1,0)}(\theta)$ and ${}^{(\varphi)}m_{(0,1)}(\theta)$
and a sum-of-squares moment
${}^{(\varphi)}m_{(2)}(\theta)=\E_\theta[(X_1^2+X_2^2)\,\varphi_\sigma(X)]$,
using a GMM-2S procedure.

The contaminated model is a \emph{scale} contamination:
$0.9\cdot t_3(\mu,I)+0.1\cdot t_3(\mu,5I)$, which inflates the
dispersion of $10\%$ of the observations without shifting the centre.
The location MLE should remain reasonable, since the $t_3$ score
for~$\mu$ is bounded (though not strongly redescending); but the scale
MLE, governed by the score $\partial\log f/\partial s$, is sensitive to
large observations and should be severely biased upward.

We compare four strategies: the WM-GMM-2S estimator; the $t_3$~MLE
(iterative reweighting); the sample mean and standard deviation
(Mean/SD); and the coordinatewise median with MAD-based scale
(Med/MAD).  Each estimator returns an estimate of $(\mu,s)$;
Table~\ref{tab:t3-2d} reports the RMSE of $\hat\mu$ (Euclidean norm of
the bias vector) and of $\hat s$ separately, over $2{,}000$
replications.

\begin{table}[ht]
\centering\footnotesize
\resizebox{\textwidth}{!}{%
\begin{tabular}{rlcccc}
\hline
& & WM-GMM-2S & MLE & Mean/SD & Med/MAD \\
$n$ & Param & RMSE & RMSE & RMSE & RMSE \\
\hline
\multicolumn{6}{l}{\textit{Clean model}} \\
100  & $\mu$ & 0.18 & 0.17 & 0.25 & 0.20 \\
     & $s$   & 0.09 & 0.08 & 0.81 & 0.17 \\
500  & $\mu$ & 0.08 & 0.08 & 0.11 & 0.09 \\
     & $s$   & 0.04 & 0.03 & 0.75 & 0.14 \\
1000 & $\mu$ & 0.06 & 0.05 & 0.08 & 0.06 \\
     & $s$   & 0.03 & 0.03 & 0.74 & 0.14 \\
\hline
\multicolumn{6}{l}{\textit{Contaminated:
  $0.9\,t_3(\mu,I)+0.1\,t_3(\mu,5I)$}} \\
100  & $\mu$ & 0.17 & 0.18 & 0.46 & 0.21 \\
     & $s$   & 0.10 & 0.19 & 2.23 & 0.29 \\
500  & $\mu$ & 0.08 & 0.08 & 0.20 & 0.09 \\
     & $s$   & 0.06 & 0.16 & 2.19 & 0.27 \\
1000 & $\mu$ & 0.06 & 0.06 & 0.14 & 0.07 \\
     & $s$   & 0.06 & 0.16 & 2.18 & 0.26 \\
\hline
\end{tabular}}%
\caption{Bivariate $t_3$ location--scale, $\mu=(0,0)$, $s=1$,
$\sigma=3$, $2{,}000$ replications.  Under the clean model, the MLE
and WM estimators are comparable.  Under scale contamination, the MLE
scale estimate does not converge (RMSE $\approx 0.16$ at $n=1000$),
while the WM scale estimator converges at the parametric rate (RMSE
$\approx 0.06$).  The Mean/SD pair is catastrophic in both settings.}
\label{tab:t3-2d}
\end{table}

Under the clean model, the MLE and WM-GMM-2S estimators are
comparable for both location and scale; the Mean/SD pair is heavily
biased in scale because the sample standard deviation estimates the
marginal standard deviation $s\sqrt{\nu/(\nu-2)}=s\sqrt{3}\approx 1.73$
rather than the scale parameter~$s=1$, producing a systematic bias of
approximately~$0.73$.  Under scale
contamination, the critical finding is in the scale parameter: the MLE
scale RMSE remains at approximately $0.16$ even at $n=1{,}000$,
indicating that the MLE scale estimate does not converge to the true
scale under contamination.  In contrast, the WM-GMM-2S scale RMSE
drops to~$0.06$ at $n=1{,}000$, converging at the parametric rate.
The Med/MAD estimator is robust but less efficient than WM-GMM-2S.

This example complements the bivariate Cauchy study: there, both the
MLE and the WM estimator were robust (from different sources); here,
the MLE scale estimator breaks down under contamination while the WM
estimator retains full performance.  The difference is structural: the
Cauchy score for location is naturally redescending, but the $t_3$
score for scale is not---it grows without bound for large observations.
The weak moment score, in contrast, inherits redescent from the kernel
regardless of the underlying family.

\subsection{Moving-atom location}
\label{subsec:sim-atom}

Finally we illustrate the non-dominated model
$P_\theta=\tfrac12\delta_\theta+\tfrac12N(0,1)$ of
Section~\ref{subsec:moving-atom}, with $\theta=1$ and Gaussian kernel
$\sigma=3$.  Because the model has no likelihood there is no MLE; we
compare the weak moment estimator (WM, $j=1$) with the sample median
and with the method-of-moments estimator $\widehat\theta=2\bar X$
(since $\E_{P_\theta}[X]=\theta/2$), over $2{,}000$ replications, under
the clean model and under $10\%$ contamination by $N(\theta+5,1)$.

\begin{table}[ht]
\centering\footnotesize
\begin{tabular}{rcccccc}
\hline
& \multicolumn{2}{c}{WM} & \multicolumn{2}{c}{Median}
& \multicolumn{2}{c}{MoM $2\bar X$}\\
$n$ & Bias & RMSE & Bias & RMSE & Bias & RMSE\\
\hline
\multicolumn{7}{l}{\textit{Clean model}}\\
$50$   & $+0.02$ & $0.27$ & $-0.03$ & $0.11$ & $+0.01$ & $0.25$\\
$100$  & $+0.00$ & $0.18$ & $-0.01$ & $0.05$ & $-0.01$ & $0.17$\\
$500$  & $+0.00$ & $0.08$ & $+0.00$ & $0.00$ & $+0.00$ & $0.08$\\
$1000$ & $-0.00$ & $0.06$ & $+0.00$ & $0.00$ & $-0.00$ & $0.05$\\
\hline
\multicolumn{7}{l}{\textit{$10\%$ contamination by $N(\theta+5,1)$}}\\
$50$   & $+0.11$ & $0.28$ & $-0.01$ & $0.05$ & $+1.12$ & $1.23$\\
$100$  & $+0.10$ & $0.21$ & $-0.00$ & $0.02$ & $+1.10$ & $1.16$\\
$500$  & $+0.09$ & $0.13$ & $+0.00$ & $0.00$ & $+1.10$ & $1.12$\\
$1000$ & $+0.09$ & $0.11$ & $+0.00$ & $0.00$ & $+1.10$ & $1.10$\\
\hline
\end{tabular}
\caption{Moving-atom location, $\theta=1$, $\sigma=3$, $2{,}000$
replications.  Under the clean model the WM RMSE matches the closed-form
prediction $\sqrt{V(\theta;\sigma)/n}$ of Table~\ref{tab:atom-analytic}
($0.058$ at $n=1000$).  Under contamination the method-of-moments
estimator $2\bar X$ is destroyed, while the weak moment estimator keeps
a small, bounded bias.  The sample median is essentially exact here
because the $\tfrac12$-atom sits at the population median---a strong
reference whenever the atom is detectable through exact ties, which the
smooth weak estimator does not require.}
\label{tab:atom-sim}
\end{table}

Two points emerge.  First, the weak estimator behaves as predicted: its
clean-model RMSE matches the closed-form asymptotic standard error of
Table~\ref{tab:atom-analytic}, and under contamination its bias remains
bounded---indeed it redescends, the contribution of an outlier at~$x$
vanishing as $|x|\to\infty$ because $x\varphi_\sigma(x)\to0$.  Second,
the method-of-moments estimator, built on the unbounded statistic
$\bar X$, is overwhelmed by a $10\%$ contamination.  The median is
near-exact in this idealised model because the atom carries half the
mass and lies at the median; the weak estimator trades that sharpness
for smoothness, differentiability, and the closed-form asymptotics
shared with the rest of the paper, and it continues to apply when the
atom is smeared by measurement noise---so that exact ties, and with them
the median's special behaviour, disappear.

\section{Discussion}
\label{sec:discussion}

The methodology developed in this paper rests on turning the
distributional framework of~\cite{LabouriauA1} into a working
programme for statistical inference.  From the single device of
replacing densities by distribution--kernel pairs and defining
expectations through the pairing $\langle T,\psi\varphi\rangle$, we
obtain a coherent approach to estimation in heavy-tailed models:
parameters can be estimated directly from weak moments, weak
characteristic functions, or weak cumulants, without reconstructing
the underlying density; and the resulting estimators are automatically
locally robust, with bounded influence function, finite gross error
sensitivity, and a redescending score---all inherited from the rapid
decay of the kernel~$\varphi$.

The central message is that the kernel simultaneously defines the
model, regularises the moments, and shapes the influence function.
This unification of moment-based and robust inference through
generalised probability is, in our view, the principal conceptual
contribution.  In the classical programme of Hampel and
Huber~\cite{Hampel1974,Huber1964,HampelRonchettiRousseeuwStahel1986},
robust estimators require careful tuning of truncation or redescent
constants; here the kernel provides this tuning as a structural
component of the model, not a post-hoc modification.

\medskip

\emph{The instrument is the tuning.}
This dual role is clarified by the viewpoint of the companion
framework~\cite{LabouriauA1}, in which $\varphi$ is not part of the
probability law but the \emph{instrument} through which the law~$T$ is
observed.  The present paper shows that this same instrument is the
robustness tuning: the gross error sensitivity, the redescent scale,
and the asymptotic efficiency are all fixed by~$\varphi$
(Corollary~\ref{cor:GES}, Proposition~\ref{prop:wm-asy}).  Choosing how
to observe the law and choosing how robust the estimator should be are
therefore one and the same act---in contrast to the classical
programme, where the model is specified first and a tuning constant is
grafted on afterwards to bound the influence function.

\medskip

\emph{Where weak estimators are most useful.}
The Cauchy location example shows that weak moment methods yield
consistent, robust estimators where no classical moment-based method
exists.  In that example, both the MLE and the WM estimator are
robust---the former because the Cauchy score is naturally redescending,
the latter because of the kernel.  The bivariate $t_3$ location--scale
example then reveals the complementary picture: the MLE scale
estimator breaks down under contamination (its RMSE does not converge),
while the WM-GMM estimator retains full performance at the parametric
rate.  This demonstrates that automatic local robustness via the kernel
is a genuine advantage whenever the likelihood score is not naturally
bounded.  More generally, the framework is most valuable in settings
where tuning is undesirable: unlike the Huber or Tukey
estimators~\cite{Huber1964,Beaton1974}, whose performance depends on
correctly calibrated constants, weak moment estimators depend on a
kernel that is part of the model specification.

\medskip

\emph{The role of the distributional CLT.}
The distributional CLT established
in~\cite[Theorem~8.11]{LabouriauA1} provides a second justification
for the asymptotic normality of weak moment estimators
(Remark~\ref{rem:dist-clt}).  Because the kernel-weighted density
$h=\varphi f$ has all moments finite, the classical CLT applies to
$h$-distributed observations, and the GMM asymptotics of
Proposition~\ref{prop:wm-asy} can be viewed as a specialisation.
This connection underscores that weak inference is not a departure
from classical statistics but a regularised extension of it.

\medskip

Several questions remain open and define a natural research programme.

\medskip

\emph{Optimal kernel selection.}
Among all positive Schwartz kernels, which minimises the asymptotic
variance subject to a gross error sensitivity constraint?  This is
the weak analogue of Hampel's optimality
problem~\cite{HampelRonchettiRousseeuwStahel1986}.  For a Gaussian
kernel $\varphi_\sigma$, the bandwidth~$\sigma$ mediates a trade-off
between efficiency (large~$\sigma$) and robustness (small~$\sigma$);
optimising this trade-off in specific parametric families is a natural
first step.

\medskip

\emph{Global robustness.}
The automatic local robustness established in Section~\ref{sec:robust}
concerns the influence function and gross error sensitivity.  A
systematic study of the maximum bias curve and finite-sample breakdown
point for weak moment estimators, as functions of the kernel and the
moment set~$\mathcal{J}$, is needed.  For Gaussian kernels, the
simulation evidence in Section~\ref{sec:simulation} suggests good
performance under moderate contamination, but a rigorous analysis
would require tools from the theory of maximum bias
curves~\cite{HuberBook}.

\medskip

\emph{Efficient GMM and the Cram\'er--Rao bound.}
For a fixed family and kernel, what is the best achievable asymptotic
variance, and how does it compare with the parametric efficiency
bound?  The GMM-2S estimator with optimal weighting is efficient
within the class of weak moment estimators, but its relationship to
the Fisher information remains to be clarified.  This is connected to
the classical theory of semiparametric
efficiency~\cite{van2000asymptotic} adapted to the weak setting.

\medskip

\emph{Multivariate extensions.}
The framework extends naturally to $\R^d$, as illustrated for
elliptically contoured models in
Section~\ref{subsec:multivariate}.
In~\cite[Theorem~7.2]{LabouriauA1}, the weak moment problem is shown
to have a unique solution for Gaussian kernels on $\R^d$ via the
completeness of the multivariate Hermite basis.  This guarantees that
the estimation programme of Section~\ref{subsec:wm} extends in
principle to $\R^d$: the weak moment sequence determines the
distribution, and empirical weak moments provide consistent estimating
equations.  The main open challenge is computational: the number of
moments of order at most $k$ in $\R^d$ grows as $\binom{k+d}{d}$, and
optimal selection of the moment set $\mathcal{J}$ becomes non-trivial.
Regularisation strategies analogous to penalised GMM may be needed.
Extensions to dependent observations (time series, regression) are
likewise natural but unexplored.

\medskip

\emph{Non-parametric density estimation.}
The reconstruction strategy of Section~\ref{subsec:reconstruction}
extends naturally to a fully non-parametric setting: observe empirical
weak expectations or the empirical weak characteristic function,
recover $g=\varphi f$, and invert $M_\varphi$ by Tikhonov or spectral
regularisation.  Sharp minimax rates under smoothness or source
conditions on~$f$, and the trade-off between tail regularisation and
ill-posedness as a function of~$\varphi$, remain to be established.
This programme sits at the intersection of classical kernel
smoothing~\cite{Cavalier2008}, deconvolution
problems~\cite{Meister2009}, and robust heavy-tail regularisation.
The distinctive feature of the weak approach is that the forward
operator $M_\varphi$ is \emph{chosen} by the analyst, rather than
dictated by an observation model---a degree of freedom that could be
exploited to optimise the bias--variance trade-off.

\medskip

\emph{Connections with other frameworks.}
The automatic robustness of weak moment estimators connects to the
broader programme of redescending $M$-estimators initiated by
Hampel~\cite{Hampel1974} and developed by Beaton and
Tukey~\cite{Beaton1974}.  The key difference is that in the classical
approach, redescent is imposed by modifying the score after the model
is specified, whereas here it arises from the model itself through the
kernel.  There is also a connection to the theory of inverse problems
in statistics~\cite{EnglHankeNeubauer1996,Cavalier2008}: the
reconstruction route of Section~\ref{subsec:reconstruction} is an
inverse problem with a \emph{chosen} forward operator, linking weak
inference to regularisation theory.  A systematic exploration of these
connections is expected to yield further insights into the interplay
between distributional modelling, robustness, and inverse methods.

\medskip

These questions define a coherent programme growing out of the present
framework, linking distributional probability, robust statistics,
inverse problems, and weighted approximation theory.


\appendix
\section{Proof of Proposition~\ref{prop:wm-asy}}
\label{app:gmm-proof}

We verify the standard regularity conditions for GMM consistency and
asymptotic normality (see, \eg
\cite{Hansen1982,NeweyMcFadden1994,van2000asymptotic}) and show that
they are satisfied automatically by the Schwartz property of the
kernel.

\medskip

\noindent\emph{Setting.}
Let $\varphi\in\mathcal{S}(\R)$ with $\varphi>0$, let
$\mathcal{J}=\{j_1,\ldots,j_K\}$ be a set of moment orders, and write
$g(x;\theta):=\bigl(x^{j_k}\varphi(x)
-{}^{(\varphi)}m_{j_k}(\theta)\bigr)_{k=1}^K$ for the moment
function vector.  The GMM estimator minimises
$Q_n(\theta):=g_n(\theta)^\top W\,g_n(\theta)$ with
$g_n(\theta):=n^{-1}\sum_{i=1}^n g(X_i;\theta)$.

\medskip

\noindent\emph{Condition~1 (Identification).}
The population criterion 
$
Q(\theta):=\E[g(X;\theta)]^\top
W$ $\E[g(X;\theta)]$
 is uniquely minimised at $\theta_0$ whenever the
Jacobian $G(\theta_0):=\partial\E[g(X;\theta)]/\partial\theta
\big|_{\theta_0}$ has full column rank and $W$ is positive definite.
This is assumed in the statement of Proposition~\ref{prop:wm-asy}.

\medskip

\noindent\emph{Condition~2 (Uniform law of large numbers).}
For each $\theta$, $|g(x;\theta)|\le|x^{j_K}\varphi(x)|
+\max_k|{}^{(\varphi)}m_{j_k}(\theta)|$.
Since $x^{j_K}\varphi(x)\in\mathcal{S}(\R)$, it is bounded and hence
$\E[|g(X;\theta)|^2]<\infty$ for every~$\theta$.
Combined with the continuity of $\theta\mapsto g(x;\theta)$ and
compactness of a neighbourhood of~$\theta_0$, this gives a pointwise
(in fact uniform) law of large numbers for $g_n(\theta)$ by
Theorem~2.4 of~\cite{NeweyMcFadden1994}.

\medskip

\noindent\emph{Condition~3 (Asymptotic normality of the moment
vector).}
By the multivariate CLT, $\sqrt{n}\,g_n(\theta_0)\xrightarrow{d}
\mathcal{N}(0,S)$, where
$S_{jk}=\mathrm{Cov}\bigl(X^{j_j}\varphi(X),\,X^{j_k}\varphi(X)\bigr)$.
Each entry of $S$ has the form
$\E_{T_\theta,\varphi^2}[x^{j_j+j_k}]
-{}^{(\varphi)}m_{j_j}(\theta)\,{}^{(\varphi)}m_{j_k}(\theta)$,
which is a weak moment with respect to the pair
$(T_\theta,\varphi^2)$ and hence exists unconditionally (since
$\varphi^2\in\mathcal{S}(\R)$ whenever $\varphi\in\mathcal{S}(\R)$).

\medskip

\noindent\emph{Condition~4 (Smoothness of the moment map).}
The assumed regularity $m(\cdot;\varphi)\in C^1(\Theta)$ ensures that
$G(\theta)$ is continuous.  Differentiation under the integral sign is
justified because $|(\partial/\partial\theta_\ell)\,g(x;\theta)|\le
|(\partial/\partial\theta_\ell)\,{}^{(\varphi)}m_{j_k}(\theta)|$,
which is integrable by the smoothness assumption.

\medskip

\noindent\emph{Conclusion.}
With Conditions~1--4 in hand, the standard GMM
theorem~\cite[Theorem~3.4]{NeweyMcFadden1994} yields consistency, and
the delta-method argument of~\cite[Theorem~3.1]{Hansen1982} yields the
asymptotic distribution
$\sqrt{n}(\hat\theta_n-\theta_0)\xrightarrow{d}\mathcal{N}(0,V)$
with
$V=(G^\top WG)^{-1}G^\top WS\,WG(G^\top WG)^{-1}$.
The optimal weight $W=S(\theta_0)^{-1}$ reduces this to the efficient
variance $V_{\mathrm{eff}}=(G^\top S^{-1}G)^{-1}$.  \qed

\end{document}